\documentclass{article}

\usepackage[inner=3cm,outer=3cm,top=3cm,bottom=3cm]{geometry} 
\usepackage{amsmath} 
\usepackage{amsthm} 
\usepackage{amssymb} 
\usepackage{mathtools} 
\usepackage{mathrsfs}
\usepackage{tikz} 
\usepackage{graphicx} 
\usepackage{tikz-cd} 
    \usetikzlibrary{cd}
\usepackage{mathrsfs} 
\usepackage{array} 
\usepackage{booktabs} 
\usepackage{setspace} 
\usepackage{xcolor} 
\usepackage{harvard} 
\usepackage{enumitem} 
\usepackage{tcolorbox} 
\usepackage{pdfpages} 
\usepackage{hyperref} 

\hypersetup{
    colorlinks=true,
    linkcolor={blue!50!black},
    citecolor={blue!50!black},
    urlcolor={blue!50!black} 
}

\theoremstyle{definition}

\theoremstyle{plain}
\newtheorem{lemma}{Lemma}
\theoremstyle{plain}
\newtheorem{theorem}{Theorem}
\theoremstyle{plain}
\newtheorem{proposition}{Proposition}
\theoremstyle{definition}

\theoremstyle{definition}
\newtheorem{claim}{Claim}
\theoremstyle{plain}

\newtheorem{corollary}{Corollary}

\title{A Topological Proof of the Gibbard-Satterthwaite Theorem}

\author{Yuliy Baryshnikov \and Joseph Root}

\date{\date{\today}}

\begin{document}

\maketitle

\abstract{We give a new proof of the Gibbard-Satterthwaite Theorem. We construct two topological spaces: one for the space of preference profiles and another for the space of outcomes. We show that social choice functions induce continuous mappings between the two spaces. By studying the properties of this mapping, we prove the theorem.}

\section{Introduction}

The Gibbard-Satterthwaite theorem is a landmark result in social choice theory and mechanism design. It delivers a striking message: the only voting rules which are not vulnerable to strategic manipulation are dictatorships. This fact not only has important implications for economics and political science, it has served as the starting point for the theory of mechanism design. The various branches of mechanism design correspond to the different ways to avoid the impossibilities implied by the Gibbard-Satterthwaite theorem. Proofs of the Gibbard-Satterthwaite theorem have primarily been combinatorial.\footnote{See for instance \citeasnoun{gibbard1973manipulation}, \citeasnoun{satterthwaite1975strategy}, \citeasnoun{barbera1983strategy}, \citeasnoun{reny2001arrow}, and \citeasnoun{sen2001another}.} An exception comes from \citeasnoun{mossel2012quantitative} who recently gave a ``quantitative" proof using analytic techniques. 

In this paper, we provide a new proof of the Gibbard-Satterthwaite theorem using tools from algebraic topology. The key idea is to view the set of preference profiles and the set of outcomes as topological spaces. The social choice function then induces a continuous map between these spaces which can be analyzed using topological techniques. By viewing the problem from this lens, we provide a richer geometric view of the impossibility result of Gibbard and Satterthwaite. 

This note contributes to the literature on topological social choice theory. This area of research began with the publication of \citeasnoun{chichilnisky1980social} which used topological techniques to establish an impossibility result for the aggregation of cardinal preferences. \citeasnoun{baryshnikov1993unifying} gave a unified proof of both Arrow's impossibility theorem and the impossibility of Chichilnisky. While a large literature has since developed studying ``topological social choice," to our knowledge, no proof of the Gibbard-Satterthwaite theorem has yet appeared. 

We start by proving the Muller-Satterthwaite theorem \cite{muller1977equivalence} which states that the only monotonic and unanimous social choice functions are dictatorships. We construct two topological spaces $N_\mathscr{P}$ and $N_A$, for the set of preference profiles and the set of outcomes respectively. A monotonic and unanimous social choice function $f$ induces a continuous map between $N_\mathscr{P}$ and $N_A$. $N_A$ is easily seen to be homotopy equivalent to a $(n-2)$-sphere. \citeasnoun{baryshnikov1993unifying} showed that, 
in dimension $n-2$, the space $N_\mathscr{P}$ has the same homology groups as the Cartesian product of $N$ spheres. We show that the homomorphism between the homologies of $N_\mathscr{P}$ and $N_A$ induced by $f$ must be a projection onto one of the coordinates, proving the theorem.

\section{Preliminaries}
Let $A=\{a_1,\dots, a_n\}$ be a finite set of alternatives. Let $P$ denote the set of linear orders on $A$.
The symbol $\succ$ will be used to denote a generic element of $P$. For a given $\succ$, let $\text{top}(\succ)$ be the $a\in A$ such that $a\succ b$ for all $b\neq a$.
Let $N\geq 1$, denote the number of agents. 
Elements of $P^{N}$ are called preference profiles. The set of preference profiles will be denoted $\mathscr{P}$. A function $f:\mathscr{P}\rightarrow A$ is called a \textbf{social choice function}.
$f$ is said to be \textbf{monotonic} if $f(\succ_1,\dots, \succ_N)=a$ and for each $i$, $\succ_i'$ is a linear order such that $a\succ_i b$ implies that $a\succ'_i b$ for all $b$ then 
$f(\succ_1',\dots, \succ_N')=a$. $f$ is said to be \textbf{unanimous} if whenever all agents top-rank some alternative $a$, $f(\succ_1,\dots, \succ_N)=a$. $f$ is said to be \textbf{dictatorial} if there is an agent $i$ such that $f(\succ_1,\dots, \succ_N)=\text{top}(\succ_i)$ for any $(\succ_1,\dots, \succ_N)$. Finally, $f$ is said to be \textbf{strategy-proof} if for every agent $i$, either $f(\succ_1,\dots, \succ_i,\dots,\succ_N)\succ_i f(\succ_1,\dots, \succ_i',\dots,\succ_N)$ or $f(\succ_1,\dots, \succ_i,\dots,\succ_N) = f(\succ_1,\dots, \succ_i',\dots,\succ_N)$ for every profile $(\succ_1,\dots, \succ_i,\dots,\succ_N)$ and every $\succ_i'$.

The aim is to prove the following theorem:

\begin{theorem}[Gibbard-Satterthwaite]
    If  $n\geq 3$, a social choice function $f$ is surjective and strategy-proof if and only if it is dictatorial. 
\end{theorem}

\section{Topological Background}

We assume familiarity with basic notions from algebraic topology.\footnote{See \citeasnoun{hatcher2002algebraic} for a good introduction.} For the reader's convenience, we briefly review a few concepts that will be central to our proof. 

An \textbf{abstract simplicial complex} is a set $V$ together with a collection of subsets $\Delta$ of $V$ such that if  $\sigma\in \Delta$ and $\sigma'\subset \sigma$ then $\sigma'\in \Delta$. Any $v\in V$ such that $\{v\}\in \Delta$ is called a \textbf{vertex} of $\Delta$. We write $V(\Delta)$ for the set of vertices of $\Delta$. If $\sigma\in \Delta$ contains $m+1$ elements, it is referred to as a \textbf{m simplex} of $\Delta$.

We will restrict attention to finite complexes where $V$ is a finite set. 

The topology of abstract simplicial complexes is derived from their so-called \textit{geometric realizations}. Given a finite abstract simplicial complex $S$, consider $\mathbb{R}^{V(S)}$, the vector space whose coordinates are indexed by the vertices of $S$. 
For any $\sigma\in S$, we can define the standard $\sigma$-simplex in $\mathbb{R}^{V(S)}$ as the convex hull of the unit vectors indexed by an element from $\sigma$. 
The \textbf{standard geometric realization} $\vert S \vert$ of the abstract simplicial complex $S$ is the union of the standard $\sigma$-simplices in $\mathbb{R}^{V(S)}$ for all $\sigma\in S$. When we refer to the topology of a simplicial complex, we mean the topology of its standard geometric realization. 

Given two simplicial complexes $S$ and $T$. A function $f:V(S)\rightarrow V(T)$ is called a \textbf{simplicial map} if for any $\sigma\in S$ we have $f(\sigma)\subset T$. A simplicial map $f:V(S)\rightarrow V(T)$ induces a continuous map between the geometric realizations of $S$ and $T$ as follows. Any $s\in \vert S \vert$ can be written as a convex combination of the vertices of $S$, i.e. $\sum_{v\in V(S)}\beta_v v$. By sending $s=\sum_{v\in V(S)}\beta_v v$ to $\sum_{v\in V(S)}\beta_v f(v)$, we get a continuous map $f:\vert S \vert \rightarrow \vert T \vert$.

Given a set $X$ and an indexed collection of its subsets $\{U_{\alpha}\}_{\alpha\in A}$ the \textbf{nerve} of $\{U_{\alpha}\}_{\alpha\in A}$ is the abstract simplicial complex $N$ where $\sigma\subset A$ is in $N$ if and only if $\bigcap_{\alpha \in \sigma}U_{\alpha}$ is nonempty.

Nerves are commonly associated with covers of a topological space. A collection of open sets $U$ in a topological space $X$ is called a \textbf{good covering} if the union of the sets in $U$ is all of $X$ and if every intersection of open sets in $U$ is either contractible (i.e., is homotopy equivalent to a point) or empty. 

\begin{lemma}[Nerve lemma]\label{lemma: nerve lemma}
Given a topological space $X$ and a good covering $U$, let $N_U$ be the nerve associated to this covering. Then $N_U$ and $X$ are homotopy equivalent.
\end{lemma}

\section{Results}

\subsection{The setup}

We first prove the Muller-Satterthwaite Theorem \cite{muller1977equivalence}. As is well-known, the Gibbard-Satterthwaite Theorem \cite{gibbard1973manipulation}\cite{satterthwaite1975strategy} can easily be deduced as a corollary.

\begin{theorem}[Muller-Satterthwaite]
Suppose $n\geq 3$. A social choice function is monotonic and unanimous if and only if it is dictatorial.
\end{theorem}

Our approach will be to construct covers of both $\mathscr{P}$ and $A$ such that a monotonic and unanimous social choice function will result in a simplicial map between the nerve of the cover of $\mathscr{P}$ and the nerve of the cover of $A$.

Let $1\leq i<j\leq n$. Define $$ U_{ij}^{+}=\{\succ \in P : a_i \succ a_j\}\hspace{2cm}U_{ij}^{-}=\{\succ \in P : a_i \prec a_j\}$$

The sets $U_{ij}^+$ and $U_{ij}^-$ cover $P$. Denote by $I_P=\{(i,j,x): 1\leq i < j \leq n, x\in \{+,-\}\}$ the index set. Let $N_P$ be the nerve of this covering, so that $J\subset I_P$ is in $N_P$ if and only if $\bigcap_{(i,j,x)\in J}U_{ij}^x$ is nonempty.

Let $\sigma=(\sigma_1,\dots,\sigma_N)\in \{-,+\}^N$. Define $$U_{ij}^{\sigma}=\{(\succ_1,\dots, \succ_N) :\hspace{0.1cm} \succ_l \in U_{ij}^{\sigma_l} \text{ for all }l\}$$ The sets $U_{ij}^{\sigma}$ cover $\mathscr{P}$. Again, denote by $I_\mathscr{P}=\{(i,j,\sigma): 1\leq i < j \leq n, \sigma\in \{+,-\}^N\}$ the index set of this covering. Let $N_\mathscr{P}$ be its nerve so that $J\subset I_\mathscr{P}$ is in $N_\mathscr{P}$ if and only if $\bigcap_{(i,j,\sigma)\in J}U_{ij}^\sigma$ is nonempty. 

Finally, the sets $U_i:=A-\{a_i\}$ ranging over $1\leq i\leq n$ cover $A$. Let $N_A$ be the nerve of this covering so that $J\subset \{1,\dots, n\}$ is in $N_A$ if and only if $\bigcap_{j\in J} U_j$ is nonempty.

\begin{lemma}\label{lemma: main simplicial map}
    If $f$ is monotonic and unanimous, it induces a well-defined simplicial map $f^s:N_\mathscr{P}\rightarrow N_A$ where $(i,j,\sigma)\mapsto i$ if $f(U_{ij}^\sigma)\subset U_i$ and $(i,j,\sigma)\mapsto j$ if $f(U_{ij}^\sigma)\subset U_j$.
\end{lemma}
\begin{proof}
    Fix $i,j,\sigma$ and let $(\succ^*_1,\dots, \succ^*_N)$ be any profile in $U_{ij}^\sigma$ where all agents rank $a_i$ and $a_j$ above all other alternatives. Let $(\succ^{**}_1,\dots, \succ^{**}_N)$ be the profile derived from $(\succ^*_1,\dots, \succ^*_N)$ by moving $a_i$ to the top of all rankings. First note that $f(\succ^*_1,\dots, \succ^*_N)\in \{a_i,a_j\}$ since otherwise by monotonicity we would have $f(\succ^{**}_1,\dots, \succ^{**}_N)\neq a_i$, a violation of unanimity. If, say $f(\succ^*_1,\dots, \succ^*_N)=a_i$ then $f(U_{ij}^\sigma)\subset A-\{a_j\}$ since any $(\succ_1,\dots,\succ_N)\in U_{ij}^\sigma$ with $f(\succ_1,\dots,\succ_N)=a_j$ would violate monotonicity as we should have $f(\succ_1,\dots,\succ_N)= f(\succ^*_1,\dots, \succ^*_N)$. Finally, it is clear that $f^s$ is a simplicial map since for any $J\in N_{\mathscr{P}}$ if $(\succ_1,\dots,\succ_N)\in \bigcap_{(i,j,\sigma)\in J}U_{ij}^\sigma$ then $f(\succ_1,\dots,\succ_N)\in \bigcap_{(i,j,\sigma)\in J}f(U_{ij}^\sigma)$.
\end{proof}

Since $f$ induces a simplicial map from $N_\mathscr{P}$ to $N_A$, it also induces homomorphisms $f_*$ from the homology groups of $N_\mathscr{P}$ to the homology groups $N_A$. Likewise it induces homomorphisms $f^*$ from the cohomology groups of $N_A$ to those of $N_\mathscr{P}$. By studying these maps, we will deduce the theorem.

\subsection{The topology of $N_\mathscr{P}$ and $N_A$}

The topology of $N_A$ is simple. 
\begin{lemma}
    The simplicial complex $N_A$ is homotopy equivalent to the $(n-2)$-sphere.
\end{lemma}

\begin{proof}
    The only intersection of the $U_i$ which is empty is the intersection of all the $U_{i}$ so the simplicial complex $N_A$ is isomorphic to the boundary of the standard $(n-1)$-simplex. 
\end{proof}

There are $n$ maximal faces of $N_A$ corresponding to $\cap_{y\neq x} U_y$ for each $x\in A$. Denote these faces $F_x$. 

The topology of $N_\mathscr{P}$ has already been described in 
\cite{baryshnikov1993unifying}. We include the calculations here for completeness.

To calculate the topology of $N_\mathscr{P}$ we first construct a manifold $M$ with a {\em good} covering $\{V_\alpha\}$. These are constructed so that the nerve $N_M$ of the covering is identical to $N_\mathscr{P}$. In this case, the homotopy types of $M$ and $N_M$ (and therefore $N_\mathscr{P}$) coincide by lemma \ref{lemma: nerve lemma}, and the focus shifts to figuring out the topology of $M$.

Let $W=\{(x_1,\dots,x_n) \in \mathbb{R}^n: \sum_i x_i = 0\}$. The manifold $M$ will be an open subset of the $N(n-1)$-dimensional vector space $V:=W^N$.

Let $1\leq i<j\leq n$. Denote the intersection of the open halfspaces with $W$ $$ K_{ij}^{+}=\{(x_1,\dots, x_n) \in W: x_i > x_j\}\hspace{2cm}K_{ij}^{-}=\{(x_1,\dots, x_n)\in W : x_i < x_j\},$$
and define, for a vector of signs $\sigma=(\sigma_1,\dots,\sigma_N)\in \{-,+\}^N$ the open polyhedral cones $$K_{ij}^{\sigma}=\{(\bar{x}^1,\dots, \bar{x}^N) :\hspace{0.1cm} \bar{x}^l \in K_{ij}^{\sigma_l} \text{ for all }l\}$$ as products of such halfspaces over voters.

Now we can introduce $M$ as the union of these cones: \[M = \bigcup_{(i,j,\sigma)\in I_\mathscr{P}}K_{ij}^\sigma.
\]
The sets $K_{ij}^{\sigma}$ are convex open polyhedra so their nonempty intersections are again convex open  polyhedra and are therefore contractible. Let $N_M$ be the nerve of the covering by $M$ of these sets so that $J\subset I_{\mathscr{P}}$ is in $N_M$ if and only if $\bigcap_{(i,j,\sigma)\in J}K_{ij}^\sigma$ is nonempty.

\begin{lemma}\label{lemma: simplicial isomorphism}
    $N_M=N_\mathscr{P}$.
\end{lemma}

\begin{proof}
Suppose $J\in N_{M}$ so that there is some $(\bar{x}^1,\dots, \bar{x}^N)$ in $K^\sigma_{ij}$ for all $(i,j,\sigma)\in J$. For each $l$, let $\epsilon_l$ be the smallest distance between any two distinct entries of $\bar{x}^l$ and let $\epsilon = \min \epsilon_l$. Choose $(\bar{y}^1,\dots, \bar{y}^N)$ so that $\vert \bar{y}^l_i-\bar{x}^l_i \vert < \epsilon/2$ for all $i,l$ and such that each entry of $\bar{y}^l$ is distinct for each $l$. Let $(\succ_1,\dots, \succ_N)$ be such that $a_i\succ_l a_j$ if and only if $\bar{x}^l_i> \bar{x}^l_j$ for all $i,j,l$. $(\bar{y}^1,\dots, \bar{y}^N)$ breaks the 
indifferences of $(\bar{x}^1,\dots, \bar{x}^N)$ so that $(\succ_1,\dots, \succ_N)\in \bigcap_{(i,j,\sigma)\in J}U_{ij}^\sigma$ and $J\in N_{\mathscr{P}}$.

Conversely, suppose $J\in N_{\mathscr{P}}$ so that there is some  $(\succ_1,\dots, \succ_N)\in U_{ij}^\sigma$ for every $(i,j,\sigma)\in J$. Let $(\bar{x}^1,\dots, \bar{x}^N)$ be any vector of utility representations. Without loss, we can rescale each $\bar{x}^i$ so that the utilites for each agent sum to zero. Then for each $(i,j,\sigma)\in J$ we have $(\bar{x}^1,\dots, \bar{x}^N)\in K^\sigma_{ij}$ and $J\in N_M$.
\end{proof}

The nerve lemma states that $N_M$ is homotopic to the manifold $M$. To further describe $M$, let $\Lambda$ be the set of functions from $\{(i,j):\hspace{0.1cm}1\leq i<j\leq n\}$ to the integers $\{1,2,\dots, N\}$. For any $\lambda\in \Lambda$, let $$R^\lambda=\{(\bar{x}^1,\dots, \bar{x}^N)\in V : \text{ for all }1\leq i<j\leq n,\hspace{0.1cm} \bar{x}^{\lambda(i,j)}_i=\bar{x}^{\lambda(i,j)}_j\}$$

\begin{claim}
    The manifold $M$ is the complement $V -\cup_{\lambda\in \Lambda}R^\lambda$.
\end{claim}

\begin{proof}
    Fix $i<j$ and $\sigma$. Any $(\bar{x}^1,\dots, \bar{x}^N)$ in $K_{ij}^{\sigma}$ is not in $\cup_{\lambda\in \Lambda}R^\lambda$ since for any $(\bar{y}^1,\dots, \bar{y}^N)$ in $\cup_{\lambda\in \Lambda}R^\lambda$ there is some $l$ such that $\bar{y}^l_i=\bar{y}^l_j$. This proves that $M\subset V -\cup_{\lambda\in \Lambda}R^\lambda.$
    Conversely, for any $(\bar{y}^1,\dots, \bar{y}^N)$ in $V -\cup_{\lambda\in \Lambda}R^\lambda$
    there is some $i<j$ such that $\bar{y}^l_i\neq\bar{y}^l_j$ for all $l$. In this case,   
    $(\bar{y}^1,\dots, \bar{y}^N) \subset \cup_{\sigma}K_{ij}^{\sigma}\subset \cup_{\sigma}\cup_{i<j}K_{ij}^{\sigma}=M$.
\end{proof}

Thus the manifold $M=V -\cup_{\lambda\in \Lambda}R^\lambda$ is the complement to the union of a collection of finitely many linear spaces in $V$, a construct often referred to as the {\em arrangement of linear subspaces} \cite{bjorner1994subspace}.


\begin{theorem}\label{thm: profile groups}
    The cohomology groups $H^k(N_M)$ are $0$ in positive dimensions less than $n-2$ and $H^{n-2}(N_M)\cong \mathbb{Z}^N$.
\end{theorem}

\begin{proof}
The mapping $\lambda$ which associates voter $k=\lambda(ij)$ to each unordered pair $(ij)$ can be interpreted as the coloring edges of the complete graph on $n$ vertices $A$ with $N$ colors.

For a color $k$, consider the graph $\Gamma_k$ with vertex set $A$ and the edges colored $k$. For a point in $R^\lambda$, the coordinates in its $k$-th component are the same if they lie in the same connected component of $\Gamma_k$ (and generically different, otherwise).

This implies that the dimension of the linear subspace $R^\lambda$ is equal to the sum over colors $k$ of the numbers of connected components of $\Gamma_k$, minus 1 (to reflect the constraint that the coordinates sum up to $0$).

Equivalently, that dimension is given by
\[
d(\lambda)=\sum_k \left(n-|E(\Gamma_k)|+h_1(\Gamma_k) - 1\right),
\]
where $|E(\Gamma_k)|$ is the number of edges of color $k$, and $h_1(\Gamma_k)$ is the rank of the $1$st homology group of $\Gamma_k$ (this follows from two definitions of Euler characteristics of a graph, via numbers of simplices and ranks of homologies).

Collecting all the terms, we obtain 
\[
d(\lambda)=(n-1)N-\binom{n}{2}+\sum_k h_1(\Gamma_k).
\]
We notice that each $h_1(\Gamma_k)$ is the sum of the ranks of $1$st homology groups over edge-connected components of $\Gamma_k$.

We claim that the $\sum_k h_1(\Gamma_k)$ is maximized when all edges are of the same color. Indeed, one can see that if one repaints a connected component in $\Gamma_k$ into some other color $l$, the $1$-cycles in that component of $\Gamma_k$ becomes  $1$-cycles of $\Gamma_l$, while new  $1$-cycles might appear, hence non-decreasing the total $h_1$. Iterating, we obtain that all the dimensions of $R^\lambda$ do not exceed those of constant $\lambda$. It remains to notice that at the last repainting step, at lease one new  $1$-cycle is introduced.

It follows that there are $N$ linear spaces of maximal dimension equal to $(n-1)N-n+1=(N-1)(n-1)$ among linear subspaces $R^\lambda$: the $k$-th subspace corresponds to all edges $(ij)$ colored in color $k$. For such constant $\lambda\equiv k$, $R^\lambda=\{(\bar{x}^1,\dots, \bar{x}^N) : \bar{x}^k_1=\cdots=\bar{x}^k_n\}$. All subspaces $R^\lambda$ with non-constant $\lambda$'s have smaller dimensions.

Denote the union of the linear subspaces $R^\lambda$ of maximal dimensions as
\[
R':=\cup_{\mbox{\tiny\bf constant } \lambda} R^\lambda\subset\cup_\lambda R^\lambda =: R.
\]

One can easily verify that $V-R'$ is the product of $N$ $(n-1)$-dimensional real spaces with the origins removed, and thus has the homotopy type of the product of $N$ spheres of dimension $(n-2)$. In particular, its cohomologies are zero in the dimensions between $0$ and $(n-2)$, and the rank of $H^{n-2}(V-R')$ is $N$.

Comparing the cohomologies of $V-R'$ and $V-R$ is easier using Alexander duality. Consider the $(N-1)(n-1)$-dimensional sphere $S=p_\infty\cup V$; and compactify $R, R'$ by adding the point $p_\infty$ at infinity to these arrangements. Then Alexander duality asserts that
\[
H^l(V-R)\cong H_{(N-1)(n-1)-l-1}(R+p_\infty,p_\infty),H^l(V-R')\cong H_{(N-1)(n-1)-l-1}(R'+p_\infty,p_\infty).
\]

Consider the long exact sequence for the triple $\{p_\infty\}\subset R'+p_\infty\subset R+p_\infty$
\[
\cdots\rightarrow H_l(R'+p_\infty,p_\infty)\rightarrow H_l(R+p_\infty,p_\infty)\rightarrow H_l(R+p_\infty,R'+p_\infty)\rightarrow H_{l-1}(R'+p_\infty,p_\infty)\rightarrow\cdots
\]
The fact that $R,R'$ has no cells in dimensions above maximal (i.e., $(N-1)(n-1)-(n-2)-1$), and that the pair $R,R'$ has no cells in the maximal dimension, one arrives at the desired conclusion.
\end{proof}

It will be useful to have a basis for $H_{n-2}(N_\mathscr{P})$ and  $H^{n-2}(N_\mathscr{P})$. To that end, we will first calculate the basis for the case when $N=1$. In this case, $N_\mathscr{P}=N_{P}$ and $M$ is simply $\mathbb{R}^n$ minus the diagonal $D$ where  $D=\{x\in \mathbb{R}^n : x_i=x_j \text{ for all }i,j\}$. 

Let $\Delta_{P}$ be the $n(n-1)-1$ simplex with the same vertices as $N_{P}$. Consider the unoriented cyclic graph $g^\circ$ with vertices $\{1,2,\dots,n\}$ and edges $(1,2),(2,3),\dots,(n-1,n),(n,1)$. Any orientation $g$ of the edges of $g^\circ$, gives a face in $\Delta_{P}$ including the vertices $(i,j,x)$ such that $x=+$ if $i\leftarrow i+1$ and $x=-$ if $i\rightarrow i+1$ in $g$. We denote this simplex as $\delta(g)$ and write its boundary as $h(g)$.

\begin{lemma}
    If $g$ is an oriented cycle, $h(g)$ generates $H_{n-2}(N_P)$. Otherwise it is zero.
\end{lemma}

\begin{proof}
Suppose that $\delta(g)=\{(1,2,\alpha_1),\dots, (n-1,n,\alpha_{n-1}),(1,n,\alpha_n)\}$. Any $n-1$ give rise to an acyclic binary relation on $\{1,2,\dots,n\}$ which can be
extended to a strict order so the corresponding face is in $N_P$. Since 
$h(g)=\partial \delta(g)$, 
it is in $H^{n-1}(N_P)$. If $g$ is acyclic, $\delta(g)\in N_P$ so that $h(g)=0$ in  $H_{n-2}(N_P)$. Next, let $g$ be an oriented cycle, for example $\alpha_k=+$ for all $k<n$ and $\alpha_n=-$. The sets $R_{12}^+,\dots,R_{n-1 n}^{+}, R_{1n}^{-}$ cover $R^n-D$ since any $x$ not in their union has $x_1\leq x_2 \leq \dots \leq x_n\leq x_1$. Then the inclusion of the subcomplex of $N_P$ whose maximal faces are the $(n-1)$-subsets of $\delta(g)$ is a homotopy equivalence and $h(g)$ is a generator for $H_{n-1}(N_P)$. 
\end{proof}

Going forward, fix an oriented cycle $\hat{g}$ and let $c$ be the associated generator from above. The universal coefficient theorem says that $H_{n-2}(N_P)\cong H^{n-2}(N_P)$ and that we can find an element $c^*$ of  $H^{n-2}(N_P)$ where $(c,c^*)=1$ and $c^*$ is a generator for $H^{n-2}(N_P)$.\footnote{$(c,c^*)$ here denotes the pairing between homology and cohomology.}

Now we're ready to give a basis for $H^{n-2}(N_P)$. Let $\Delta_{\mathscr{P}}$ be the simplex with the same vertices as $N_{\mathscr{P}}$. Any collection $(g_1,\dots, g_N)$ of orientations of $g^0$ gives a face in $\Delta_{\mathscr{P}}$ including the $(i,j,\sigma)$ where $\sigma_l = +$ if $i\rightarrow i+1$ in $g_l$ and $\sigma_l = -$ if $i\leftarrow i+1$ in $g_l$. We denote this simplex as $\delta(g_1,\dots, g_N)$ and write its boundary as $h(g_1,\dots, g_N)$.

For each $l$, let $p_l$ be the simplicial map from $N_{\mathscr{P}}$ to $N_P$ sending $(i,j,\sigma)$ to $(i,j,\sigma_l)$. Let $p_l^*$ be the induced homomorphism $H^{n-2}(N_P)\rightarrow H^{n-2}(N_\mathscr{P})$.

For each $l$, fix some $(g^l_1,\dots, g^l_N)$ where $g^k_l$ is acylic if $k\neq l$ and $g^l_l=\hat{g}$ let $h_l=h(g^l_1,\dots, g^l_N)$.

\begin{lemma}
    The collection $\{h_1,\dots, h_N\}$ is a basis for $H_{n-2}(N_\mathscr{P})$ and $\{p_1^*(c^*),\dots, p_N^*(c^*)\}$ is the dual basis for $H^{n-2}(N_\mathscr{P})$.
\end{lemma}

\begin{proof}
    From the universal coefficient theorem,  $H_{n-2}(N_\mathscr{P})\cong H^{n-2}(N_\mathscr{P})$. For any $k$ and $l$, $(h_k,p_l^*(c^*))$ is $1$ if $k=l$ and is zero otherwise. It is a simple exercise to verify that both must then be bases. 
\end{proof}
This allows us to conclude the following.
\begin{corollary} \label{cor:homologous orientations}
    For any two tuples of orientations $(g_1,\dots,g_N)$ and $(g_1',\dots,g_N')$ where for some $k$, (1) $g_k=g_k'$ (2) both $g_k$ and $g_k'$ are acyclic and (3) $g_l$ and $g_l'$ are acyclic for all $l\neq k$ then $h(g_1,\dots,g_N)=h(g_1',\dots,g_N')$ in $H_{n-2}(N_\mathscr{P})$.
\end{corollary}

\subsection{The Topology of a Social Choice Function}

Let $d^*$ be a generator of $H^{n-2}(N_A)$. 

\begin{proposition}\label{prop: dictators in homology}
    Suppose $f$ is monotonic and unanimous. $(f_{\star}(h_l),d^*)=1$ if and only if $l$ is a dictator of $f$. Otherwise it is zero.
\end{proposition}

\begin{proof}
    Let $(1,2,\sigma^1),\dots (n-1,n,\sigma^{n-1}),(1,n,\sigma^n)$ be the vertices of $\delta(g^l_1,\dots, g^l_N)$. $f^s$ maps each of these vertices onto a vertex of $N_A$. If $f^s$ maps any two vertices of  $\delta(g^l_1,\dots, g^l_N)$ to the same vertex of $N_A$ then $(f_{\star}(h_l),d^*)=0$ and $l$ is not a dictator since in this case there is some $a\in A$ which is not chosen even when $l$ top-ranks it. 

    Conversely, suppose that $f^s$ is injective on the vertices of $\delta(g^l_1,\dots, g^l_N)$. Consider some $g'$ derived from $g^l_k$ by swapping one of the arrows without forming a cycle. All but one of the vertices of $\delta(g^l_1,\dots, g',\dots g^l_N)$ are vertices of $\delta(g^l_1,\dots, g^l_N)$. Let $h_l'=h(g^l_1,\dots, g',\dots g^l_N)$. Corollary \ref{cor:homologous orientations} implies $h_l=h_l'$. The new vertex of $\delta(g^l_1,\dots, g',\dots g^l_N)$ must mapped to the same vertex of $N_A$ as the one it replaced in $\delta(g^l_1,\dots, g^l_N)$ since otherwise, $(f_\star(h_l),d^*)\neq 0$ and $(f_\star(h_l),d^*)= 0$. Repeating this process, changing one arrow at a time, for any $i$ we can reach the a tuple of orientations where $i\leftarrow i+1$ for all orientations.
\end{proof}

Finally, consider $h(\hat{g},\dots,\hat{g})$. Unanimity implies that $(f_\star(h(\hat{g},\dots,\hat{g})),d^*)=1$. Together with proposition \ref{prop: dictators in homology}, we see that there must be exactly one dictator, concluding the proof of the Muller-Satterthwaite theorem.

To prove the Gibbard-Satterthwaite theorem, we need the following simple fact, well-known in social choice theory (see for example \citeasnoun{muller1977equivalence}).

\begin{proposition}
    A social choice function is monotonic and unanimous if and only if it is surjective and strategy-proof.
\end{proposition}

\begin{proof}
    A unanimous social choice function is clearly surjective. A monotonic social choice function is also strategy-proof: if there were some $f(\succ_i,\succ_{-i})\prec_i f(\succ_i',\succ_{-i})$, letting $\succ_i^*$ be the preference derived from $\succ_i$ by pushing $f(\succ_i',\succ_{-i})$ to the top, monotonicity implies $f(\succ_i,\succ_{-i})=f(\succ_i^*,\succ_{-i})= f(\succ_i',\succ_{-i})$, a contradiction.

    For the converse, suppose that $f$ is strategy-proof and surjective. Let $f(\succ_1,\dots, \succ_N)=a$ and for each $i$, $\succ_i'$ is a linear order such that $a\succ_i b$ implies that $a\succ'_i b$ for all $b$. We have that $f(\succ_1,\succ_2,\dots, \succ_N)=f(\succ_1',\succ_2,\dots, \succ_N)=f(\succ_1',\succ_2',\dots, \succ_N)=\cdots = f(\succ_1',\succ_2',\dots, \succ_N')$, so that $f$ is monotonic. If $f$ is strategy-proof and surjective it is also unanimous since for any $a$ there is some $(\succ_1,\dots,\succ_N)$ such that $f(\succ_1,\dots,\succ_N)=a$ and for any profile $(\succ'_1,\dots,\succ'_N)$ where all agents top-rank $a$, we have $f(\succ'_1,\dots,\succ'_N)$ by monotonicity.
\end{proof}

This proves the Gibbard-Satterthwaite Theorem. 

{\footnotesize
\bibliography{topologicalGS}
\bibliographystyle{econometrica}
}

\end{document}